\documentclass[runningheads]{llncs}
\usepackage{balance}
\usepackage{graphicx}
\usepackage{color}
\usepackage{ifthen}
\usepackage{xspace}
\usepackage{algorithm}
\usepackage{algorithmic}
\usepackage{amsmath, amssymb, amsfonts}
\usepackage[utf8]{inputenc}
\usepackage{fullpage}
\usepackage{hyperref}
\usepackage{epstopdf}
\usepackage[T1]{fontenc}
\begin{document}
\title{Dynamics and Convergences for Markov Coevolutionary Opinion Formation Games in Dynamic Social Networks}
%
%
\author{Po-An Chen\inst{1} \and
Chi-Jen Lu\inst{2} \and
Chuang-Chieh Lin\inst{3} \and
Jim Shi\inst{4} \and
Chih-Chieh Hung\inst{5}}
\authorrunning{P.-C. Chen et al.}
%
\institute{Institute of Information Management, National Yang Ming Chiao Tung University, Taiwan\\
\email{poanchen@nycu.edu.tw}\\
\and
Institute of Information Science, Academia Sinica, Taiwan\\
\email{cjlu@iis.sinica.edu.tw}\\
\and 
Department of Computer Science and Engineering, National Taiwan Ocean University, Taiwan\\
\email{josephcclin@mail.ntou.edu.tw}\\
\and
Tuchman School of Management, New Jersey Institute of Technology, USA\\
\email{jshi@njit.edu}\\
\and 
Department of Management Information and Information Systems, National Chung Hsin University, Taiwan
\\\email{smalloshin@nchu.edu.tw}
}

\maketitle              
\begin{abstract}
While deterministic variants of the coevolutionary opinion formation games such as the $K$-Nearest Neighbor ($K$-NN) game (\cite{bhawalkar}) in a dynamic social network environment can sometimes be shown to stabilize using potential functions or localized smoothness arguments, introducing stochasticity fundamentally changes the mathematical landscape. In the ``$K$-NN Markov game", network topologies evolve via a time-varying, randomized selection process. Proving whether such a system, as a special case of general-sum \emph{Markov games}, converges to an equilibrium is a profoundly non-obvious and challenging theoretical question.
 
Multiagent reinforcement learning has been shown to derive Nash (minimax) equilibria in two-player zero-sum Markov games and Markov potential games (along with some price-of-anarchy types of results). In recent work, optimistic dynamics are shown to converge to correlated equilibria in general-sum Markov games while the price-of-anarchy bounds are unknown.     
We thus analyze playing specific no-regret algorithms in general-sum Markov games for convergence to a stricter set than correlated equilibria.   

We integrate the convergence analysis techniques from multi-agent reinforcement learning in works of Wei et al. and online learning in a recent work of Anagnostides et al.. Specifically in (general-sum) Markov games, since the regret of the optimistic gradient ascent algorithm would have extra positive terms coming from Q-values, taking care of these terms requires non-trivial extra work setting an appropriate range of our learning rate and deriving the threshold on the number of iterations for convergence or a bounded price of anarchy, significantly different from those in the assumption in a main technical theorem of Anagnostides et al.. 
We analyze a weaker sense of convergences to approximate Nash equilibria by playing optimistic gradient ascents in general-sum Markov games.   
Specific no-regret algorithms beyond zero-sum Markov games converge not only to correlated equilibria, but also to approximate Nash equilibria (a stricter set than correlated equilibria) or a bounded price of anarchy.

\end{abstract}
\section{Introduction}
In classical models of social networks, opinion dynamics and network topology are frequently analyzed in isolation \cite{degroot1974,friedkin1990}\cite{bindel:kleinberg,bhawalkar,chen2016}. Either opinions change over a fixed graph structure, or the network adapts based on static attributes of the nodes. However, real-world social environments exhibit a \textit{coevolutionary feedback loop}: individuals adjust their expressed opinions to better align with their friends, while simultaneously rewriting their social ties to gravitate toward those with compatible views \cite{bhawalkar}. When formalized as a game, this interaction yields a complex, dual-layered dynamical system where continuous actions (expressed opinions) shape discrete states (network topologies), which in turn dictate the payoffs that drive subsequent actions.

While deterministic variants of these coevolutionary games such as the $K$-Nearest Neighbor ($K$-NN) game (\cite{bhawalkar}) can sometimes be shown to stabilize using potential functions or localized smoothness arguments, introducing stochasticity fundamentally changes the mathematical landscape. On the other hand, Wang et al. proposed a language-based simulation framework in which agents interact through text, update their beliefs through reasoning, and adjust their network connections based on opinion similarity \cite{wang2025}. Their results showed that LLM agents reproduce polarization patterns more faithfully  than traditional numerical models like the Friedkin-Johnsen models~\cite{friedkin1990}. In the $K$-NN Markov game, network topologies evolve via a time-varying, randomized selection process. Rather than connecting strictly to the absolute closest peers, agents sample a neighborhood based on localized probabilities where ``roughly nearby'' opinions have higher weights.

Proving whether such a system converges to an equilibrium is a profoundly non-obvious and challenging theoretical question for several key reasons:
\begin{enumerate}
\item \textbf{The Coevolutionary Feedback Loop:} A slight shift in a single player's expressed opinion can alter the transition probabilities across the entire network. This creates a highly coupled feedback loop where continuous actions dynamically reshape the state space boundaries, potentially driving the system into persistent oscillations or chaotic trajectories.
\item \textbf{Stochastic Disruption of Stability:} In a deterministic system, once an optimal or near-optimal topology is found, the system can remain locked there. In this stochastic setting, the randomized neighbor selection ensures that there is always a non-zero probability of ``suboptimal'' link formation. These random structural shocks can repeatedly dislodge the system from localized opinion consensus, forcing agents to continuously re-optimize their expressions.
\item \textbf{Non-Linear State Transitions:} The probability distribution governing the next network topology is a non-linear product of exponential softmax weights, parameterized by the continuous action profiles of all agents. Standard Markov game architectures typically assume fixed or linear transition geometries. Here, the transitions are driven by continuous distances, rendering conventional matrix-based convergence proofs inapplicable.
\end{enumerate}

Understanding whether this system stabilizes requires analyzing the existence and convergence to certain equilibrium concepts. If the expressed opinions and time-varying topologies do stabilize (either pointwise or in distribution), it implies that structural and ideological predictability can emerge out of decentralized, noisy interactions. Conversely, if they do not, it reveals that the mere act of randomized neighbor tracking can permanently destabilize a society, leading to endless cycles of polarization and realignment. To rigorously analyze these dynamics, we first formalize this system within the explicit mathematical syntax of a Markov game in Section~\ref{sec:k-nn}.

\subsection{Convergences to Equilibria in Games}
\label{sec:problems}



When each participant in a repeated game uses a \emph{no-regret} learning algorithm to choose actions, it is well known that the empirical frequency of the participants will converge to coarse correlated equilibria (CCE), 
meaning that the ``time-averaged'' strategy will converge to more relaxed equilibria (a larger set).
Plenty of work has been done showing that when using a specific no-regret algorithm, convergence of plays to equilibria can be shown and better quality of outcomes in terms of so-called ``the price of anarchy'' (POA) can be achieved in some classes of games.

Series of papers affirmatively answers one of the questions that multiagent learning wants to answer: in ``what types of games'' and ``what types of no-regret learning algorithms'' can strategy profiles reach convergence? These results confirm that a large class of learning algorithms can converge to an approximate equilibrium in a large class of games, and this is to converge to a smaller set than CCE while further ensuring an upper bound on the ratio of its social cost to the optimal social cost (i.e., price of anarchy). 
Can similar or other specific no-regret learning algorithms be used in more general (or different) classes of games such as \emph{Markov} games (also called stochastic games) \cite{shapley} for convergences with well-bounded convergence rates?

The framework of Markov games is often used to model strategic interactions among agents in a stochastic environment (also modeled as an extra agent). It can be seen as a generalization of Markov Decision Processes (MDPs) from a single agent to multiple agents. Multiagent reinforcement learning has been shown to converge to Nash (or minimax) equilibria in two-player \emph{zero-sum} Markov games~\cite{wei2021last} and Markov \emph{potential} games ~\cite{leonardos,chen:zhang} (along with
some price-of-anarchy types of results)~\cite{chen:zhang}, where 
a Markov potential game has a corresponding potential function for each state~\cite{leonardos}. 

There are several real-world applications in different domains: in competitive pricing and market entry, zero-sum Markov games model how firms adapt strategies over time under uncertainty; in adversarial supply chain management, zero-sum stochastic games model such adversarial interactions with sequential decisions; in negotiation and contracting in adversarial contexts, zero-sum Markov games can simulate multi-round negotiations where each side adjusts based on observed concessions or aggressions. 
However, a recent study on inventory management~\cite{liu:hu} finds that policies learned via multi-agent deep reinforcement learning perform best in practice when agents are neither fully self-interested nor fully system-focused.
In a recent work, multiagent reinforcement learning in \emph{general-sum} Markov games guarantees convergence to correlated equilibria~\cite{cai:luo}. 

\subsubsection{Our Results.}

Inspired by the convergence result of playing the Optimistic Gradient Descent Ascent (OGDA) algorithm in two-player zero-sum (Markov) games~\cite{wei2021last,wei2021linear}, we study no-regret dynamics in general-sum Markov games:\\
\emph{
Combining the convergence analysis techniques from multi-agent reinforcement \cite{wei2021last} and online learning \cite{anagnostides}, we then analyze a weaker sense of convergences to approximate Nash equilibria in general-sum Markov games by playing optimistic gradient ascents, where the result can be readily extended for multiple players.} 

The techniques that we use in comparison to the previous work for general-sum games can be highlighted as follows. In (general-sum) Markov games, the regret of the optimistic gradient ascent algorithm would have extra positive terms, coming from Q-values, that need to be offset by part of the negative terms similar to those in the proof of \cite{anagnostides}[Theorem~A.12], from which the proof of \cite{anagnostides}[Theorem~A.17] follows. 
In our case, this results in assumption~(iii) required in \cite{anagnostides}[Theorem~A.17] unsatisfied since the range of our learning rate (in terms of number of iterations for general-sum two-player games) has to be different from a constant learning rate for general-sum two-player games in assumption~(iii) (or a learning rate in terms of the number of agents for general-sum multi-player games in assumption~(iii)). 
Moreover, the threshold on the number of iterations for convergence or the bounded price of anarchy in (iii) is now different (and is in relation to the number of states, which is not in the threshold on the number of iterations in (iii)).

\subsection{Previous and Other Related Work}
\subsubsection*{Zero-sum bilinear games}
One can formulate the well-known constrained saddle-point problem or equivalently the minimax problem in a zero-sum (bilinear) game as follows:\\ $\min_{\mathbf{x}\in\mathcal{X}} \max_{\mathbf{y}\in\mathcal{Y}}f(x,y)$ where $\mathcal{X}$
and $\mathcal{Y}$ are compact convex sets, and $f$ is a continuous differentiable function that is convex in $\mathbf{x}$ for any fixed $\mathbf{y}$ and concave in $\mathbf{y}$ for any fixed $\mathbf{x}$. 

The set of minimax optimal strategies is denoted by $\mathcal{X}^* = \arg\min_{\mathbf{x}\in\mathcal{X}}\max_{\mathbf{y}\in\mathcal{Y}}f(\mathbf{x},\mathbf{y})$, and the set
of maximin optimal strategies is denoted by $\mathcal{Y}^* = \arg\max_{\mathbf{y}\in\mathcal{Y}}\min_{\mathbf{x}\in\mathcal{X}}f(\mathbf{x},\mathbf{y})$. Note that $\mathcal{X}^*$ and $\mathcal{Y}^*$ are compact and convex, and any pair $(\mathbf{x}^*,\mathbf{y}^*) \in \mathcal{X}^*\times\mathcal{Y}^*$ is a Nash  (minimax) equilibrium that satisfies $f(\mathbf{x}^*,\mathbf{y}) \leq f(\mathbf{x}^*,\mathbf{y}^*) \leq f(\mathbf{x},\mathbf{y}^*)$ for any $(\mathbf{x},\mathbf{y}) \in \mathcal{X}\times\mathcal{Y}$.
For a point $(\mathbf{x},\mathbf{y}) \in \mathcal{X}\times\mathcal{Y}$, we further define 
\[F(\mathbf{x},\mathbf{y}) = (\nabla_\mathbf{x}f(\mathbf{x},\mathbf{y}),-\nabla_\mathbf{y}f(\mathbf{x},\mathbf{y})).\] The goal is to find a point $\mathbf{x},\mathbf{y}\in\mathcal{X}\times\mathcal{Y}$ that is close to the set of Nash (minimax) equilibria $\mathcal{X}^*\times\mathcal{Y}^*$.

It is well-known (e.g., \cite[Section~3]{abernethy:bartlett}) that two players playing $o(T)$-regret algorithms $\mathcal{L}_\mathcal{X}$ and $\mathcal{L}_\mathcal{Y}$, respectively, in a zero-sum game with a cost function $f: \mathcal{X}\times \mathcal{Y}\to \mathbb{R}$ of the form $f(\mathbf{x},\mathbf{y}) = \mathbf{x}^{\top} M\mathbf{y}$ for some $M\in \mathbb{R}^{n\times m}$ give a version of minimax equilibrium.
This standard technique and result have been existing for playing generic no-regret algorithms in a zero-sum bilinear game, which is restated in the following.
\begin{theorem}[Corollary 3 of \cite{abernethy:bartlett}] \label{thm:minmax}
For compact convex sets $\mathcal{X} \subset \mathbb{R}^n$ and $\mathcal{Y} \subset \mathbb{R}^m$ and any biaffine function\footnote{A biaffine function $f: \mathcal{X} \times \mathcal{Y} \mapsto \mathbb{R}$ satisfies $f(\alpha \mathbf{x} + (1-\alpha)\mathbf{x}',\mathbf{y}) = \alpha f(\mathbf{x},\mathbf{y}) + (1-\alpha)f(\mathbf{x}',\mathbf{y})$ and
$f(\mathbf{x},\alpha \mathbf{y} + (1-\alpha)\mathbf{y}') = \alpha f(\mathbf{x},\mathbf{y}) + (1-\alpha)f(\mathbf{x},\mathbf{y}')$ for every $0 \leq \alpha \leq 1$, $\mathbf{x},\mathbf{x}'\in \mathcal{X}$ and $\mathbf{y},\mathbf{y}'\in \mathcal{Y}$.}
$f : \mathcal{X} \times \mathcal{Y} \mapsto \mathbb{R}$, we have
\begin{eqnarray}
\min_{\mathbf{x}\in \mathcal{X}}\max_{\mathbf{y}\in \mathcal{Y}}f(\mathbf{x},\mathbf{y})=\max_{\mathbf{y}\in \mathcal{Y}}\min_{\mathbf{x}\in \mathcal{X}}f(\mathbf{x},\mathbf{y}).
\end{eqnarray}
\end{theorem}

However, simple algorithms such as Gradient Descent Ascent
(GDA) and Multiplicative Weights Update (MWU) are known to cycle and fail to converge even in bilinear cases (e.g.,~\cite{bailey}). Nonetheless, in~\cite{wei2021linear} the authors derived the last-iterate convergence results for the Optimistic Gradient Descent Ascent (OGDA) algorithm (OGDA) and the Optimistic Multiplicative Weights Update (OMWU) in the constrained setting.

\subsubsection*{Beyond zero-sum and potential games}
In recent years, this line of research has been focused on finding ``which classes of games'' and ``which types of learning algorithms'' can enhance the convergence results. There are many potential types of algorithms and applicable classes of games. New developments such as \cite{hsieh,golowich} focus on ``monotonic games''~\cite{golowich}, which are not necessarily potential games, using a class of ``optimistic gradient descents'' modified from the classic gradient descents, which will be introduced in Section~\ref{sec:ogda}. The \emph{optimistic gradient descent} algorithm, i.e., intuitively the algorithm of ``two gradient updates", achieves the convergence of the strategy profile in the ``time-average'' and ``standard (that is, last-iterate)'' ways~\cite{golowich}.

\subsubsection*{Single-call stochastic extra-gradient methods in monotone games}
The game~$\mathcal{G}=(N,(\mathcal{K})_{i=1}^n,(f_i)_{i=1}^n))$ is monotone if for
all $\mathbf{x}', \mathbf{x} \in\mathcal{K}$, it holds that $\langle F_\mathcal{G}(\mathbf{x}')-F_\mathcal{G}(\mathbf{x}), \mathbf{x}' - \mathbf{x}\rangle\geq 0$. In such a case, we say also that $F_\mathcal{G}$ is a monotone operator.
The problem that is closely related to the \emph{monotonic game} is so-called ``variational inequalities (VI)''. Let $\mathcal{K}\subseteq\mathbb{R}^d$ and let $V:\mathbb{R}^d\rightarrow \mathbb{R}^d$ be a single-valued operator. In its most general form, the \emph{variational inequality (VI)} problem associated with~$V$ and $\mathcal{K}$ is as follows~\cite{hsieh}:
$$\hbox{find }\mathbf{x}^*\in\mathcal{K}\hbox{ that satisfies } \langle V(\mathbf{x}^*),\mathbf{x}-\mathbf{x}^* \rangle\geq 0\forall \mathbf{x}\in\mathcal{K}.$$
Under this general \emph{variational inequality problem}, if $f$ is a smooth loss function on $\mathcal{K}=\mathbb{R}^d$ and $V=\nabla f$, then $\mathbf{x}^*\in\mathcal{K}$ is the solution to VI if and only if $\nabla f(\mathbf{x}^*)=0$. Such a monotonic VI problem is \emph{not} a monotone game $\mathcal{G}=(N,(\mathcal{K})_{i=1}^n,(f_i)_{i=1}^n))$, but a loss minimization problem with only one participant.
In comparison, note that the monotonicity of operator $V$ is as follows:
$\langle V(\mathbf{x}')-V(\mathbf{x}),\mathbf{x}'-\mathbf{x}\rangle\geq 0,\forall \mathbf{x},\mathbf{x}'\in\mathbb{R}^d,$
which is equal to the convexity of~$f$.

When $V$ is ``smooth'' and ``monotone'', the ``extra-gradient (EG)'' algorithm updating decisions~$x$ can guarantee the convergence of the \emph{average} decision and its convergence rate~$1/T$ when $T$ is the number of iterations \cite{hsieh}. We will introduce the EG algorithm and other algorithms invoking an oracle only once to obtain gradient vectors and to approximate EG, including ``optimistic gradients (OG)''.

\subsubsection*{Last iterate slower than averaged iterate in smooth convex-concave saddle point problems}
Consider a two-player zero-sum monotone game $\mathcal{G}$ so $V(\mathbf{x})$ is $F_\mathcal{G}(\mathbf{x})$. From the average convergence result of~\cite{hsieh}, $\|F_\mathcal{G}(\mathbf{\mathbf{x}}^{(t)})\|$ is chosen as the \emph{approximate potential} function. 
When time $t^*$ is large enough such that $\|F_\mathcal{G}(\mathbf{x}^{(t^*)})\|\leq O(1/\sqrt{T})$ is satisfied, its value will only increase by a small value every time: that is, after $t\geq t^*$ for some~$t^*$ large enough, $\|F_\mathcal{G}(\mathbf{x}^{(t+1)})\|\leq (1+O(1/T))\|F_\mathcal{G}(\mathbf{x}^{(t)})\|.$

\section{Preliminaries}
Some standard equilibrium concepts for strategic games are defined and restated in Appendix~\ref{app:a}.
There are some basic notions about Markov games reviewed and modeling stochastic K-NN games as a special case of Markov games in the following.
\subsubsection{Markov games.}
We consider a Markov game~$\langle\mathcal{N},\mathcal{S},\mathcal{A},P,\vec{r}\rangle$ with $N$ agents whose indices are denoted by $i \in \mathcal{N} = \{1,...,N\}$, state space $\mathcal{S}$, action space $\mathcal{A} = \{\mathcal{A}^1,...,\mathcal{A}^N\}$,
transition function $P:\mathcal{S}\times\mathcal{A}\rightarrow\Delta(\mathcal{S})$, reward functions~$\vec{r} = \{r^i\}_{i\in\mathcal{N}}$ with $r^i:\mathcal{S}\times\mathcal{A}\rightarrow\mathbb{R}$ for each $i\in\mathcal{N}$, and initial state distribution~$\mu\in\Delta(\mathcal{S})$ \cite{agarwal:kakade,daskalakis:foster}. Assume that each agent observes the state $s\in\mathcal{S}$. Under this, we consider \emph{joint policies (i.e., policy profile)},
$\pi: \mathcal{S}\rightarrow\times_{i\in\mathcal{N}}\Delta(\mathcal{A}^i)$, factored as the product of individual policies $\pi^i: \mathcal{S}\rightarrow\Delta(\mathcal{A}^i)$.
For a stationary policy profile $\pi = (\pi_i)_i$ and an initial state~$s$, the expected discounted value that agent~$i$ obtains is
\begin{eqnarray*}
V_\pi^i(s) &=& \mathbf{E}\Bigl[\sum_{t=1}^\infty\gamma^{t-1}r^i(s_t,a_t)|s_1 = s,a_t\sim\pi(s_t),s_{t+1} \sim P(\cdot|s_t,a_t)\forall t>1\Bigr],
\end{eqnarray*}
where $a_t$ is selected according to the policy profile~$\pi$.
We also denote $V^\mu_i(\pi) = \mathbb{E}_{s\sim\mu}[V^s_i(\pi)]$ if the initial state is random and follows distribution~$\mu$.

\paragraph{Nash equilibrium policy profile.} We focus on the solution concept of \emph{approximate Nash policy profile}, as formally defined below.

\begin{definition}[$\varepsilon$-Nash equilibrium (policy profile)] 
The \emph{Nash-gap} of a policy~$\pi$ is defined as
\[\text{Nash-gap}(\pi) = \max_i\Bigl(\max_{\pi'_i}V_{\pi'_i,\pi_{-i}}^i(\mu)
-V_\pi^i(\mu)\Bigr).\]
A policy profile~$\pi = (\pi_1,...,\pi_N)$ is an $\varepsilon$-Nash policy profile if $\text{Nash-gap}(\pi)\leq\varepsilon$. 
\end{definition}

\paragraph{Correlated equilibrium policy profile.} A correlated equilibrium is a joint policy where no
player can increase their value by any strategy modification.
Formally, it is defined as
\begin{definition}[Correlated Equilibrium (Policy Profile) \cite{cai:luo}] 
A joint policy $\pi$ is a correlated equilibrium (CE) if
\[\max_i\max_{\phi_i}V^i_{(\phi_i\diamond\pi_i)\odot\pi_{-i}}(s) - V^i_\pi(s) \leq 0.\] 
A joint policy~$\pi$ is an $\varepsilon$-approximate CE if \[CEGap(\pi) = \max_i\max_{\phi_i} V^i_{(\phi_i\diamond\pi_i)\odot\pi_{-i}}(s) - V^i_\pi(s)\leq\varepsilon,\]
where the modified policy is denoted by $\phi_i\diamond\pi_i$.
\end{definition}

\paragraph{Coarse Correlated equilibrium policy profile.} A coarse correlated equilibrium is a joint policy where no player can increase their value by playing any other independent policy. Formally, it is defined as
\begin{definition}[Coarse Correlated Equilibrium (Policy Profile) \cite{cai:luo}] 
A joint policy~$\pi$ is an $\varepsilon$-approximate coarse correlated equilibrium if 
\[\max_i\max_{\pi'_i}V^i_{\pi'_i\times\pi_{-i}}(s)-V^i_\pi(s)\leq\varepsilon.\]
\end{definition}

\paragraph{Regret in Markov games.} 
The performance of the learner~$i$ (i.e., player~$i$) is measured against the best stationary policy in hindsight given the others' policies fixed, giving rise to the expected regret:
\begin{definition}[Regret of Player~$i$]
\begin{eqnarray*}
R^i_T(s)&=&\max_{\pi'_i}V_{\pi'_i,\pi_{-i}}^i(s)-V_\pi^i(s)\\ 
&=&\max_{\pi'_i}\mathbf{E}\Bigl[\sum_{t=1}^T\gamma^{t-1}r_i(s_t,a_t)|s_1 = s,a_t\sim\pi'_i(s_t),s_{t+1}\sim P(\cdot|s_t,a_t)\forall t\geq 1\Bigr]\\ 
&&-\mathbf{E}\Bigl[\sum_{t=1}^T\gamma^{t-1}r_i(s_t,a_t)|s_1 = s,a_t\sim\pi_i(s_t),s_{t+1} \sim P(\cdot|s_t,a_t)\forall t\geq 1\Bigr]
\end{eqnarray*}
for any~$s$.
\end{definition}

\subsection{Proof of the Markovian Property for Network Topologies}
To prove that the network topologies $(G^{(t)})$ as states satisfy the Markov property, we must show two things:
\begin{itemize}
\item The Controlled Markov Property: Given the entire history of the game, the transition probability to the next network topology depends only on the current state and the current joint action profile.
\item The Autonomous Markov Chain Property: Under stationary Markov policies, the sequence of topologies $\{G^{(t)}\}_{t=0}^\infty$ forms a classical Markov chain where the future topology depends exclusively on the current topology.
\end{itemize}

Let the history of the game up to time step $t$ be denoted by the continuous sequence of observed topologies and joint actions:
\[
H^{(t)} = \left(G^{(0)}, \mathbf{z}^{(0)}, G^{(1)}, \mathbf{z}^{(1)}, \dots, G^{(t-1)}, \mathbf{z}^{(t-1)}, G^{(t)}\right)
\]
where $G^{(\tau)} \in \mathcal{S}$ represents the network state and $\mathbf{z}^{(\tau)} \in \mathcal{A}$ represents the joint expressed opinion vector at time $\tau$.

\subsubsection*{Step~1: Proving the Controlled Markov Property}
To establish that the state space is conditionally Markovian under joint actions, we must verify that for any next state $G^{(t+1)} \in \mathcal{S}$:
\[
P\left(G^{(t+1)} \;\middle|\; H^{(t)}, \mathbf{z}^{(t)}\right) = P\left(G^{(t+1)} \;\middle|\; G^{(t)}, \mathbf{z}^{(t)}\right)
\]

Recall that a state $G^{(t+1)}$ is uniquely defined by the collection of individual node neighbor choices: $G^{(t+1)} = \left(S_1^{(t+1)}, \dots, S_n^{(t+1)}\right)$. According to the problem specification, each node $i \in \mathcal{N}$ conditionally independently samples its neighborhood $S_i^{(t+1)}$ based on the exponential softmax function of expressed opinions. Thus, the joint transition probability factorizes as:
\begin{equation}\label{eq:factorization}
P\left(G^{(t+1)} \;\middle|\; H^{(t)}, \mathbf{z}^{(t)}\right) = \prod_{i=1}^n P\left(S_i^{(t+1)} \;\middle|\; H^{(t)}, \mathbf{z}^{(t)}\right)
\end{equation}

By examining the neighborhood choice probability for an individual agent $i$:
\[
P\left(S_i^{(t+1)} \;\middle|\; H^{(t)}, \mathbf{z}^{(t)}\right) = \frac{\prod_{j \in S_i^{(t+1)}} \exp\left(-\beta \left|z_j^{(t)} - s_i\right|\right)}{\sum_{\substack{S' \subset \mathcal{N} \setminus \{i\} \\ |S'| = K}} \prod_{k \in S'} \exp\left(-\beta \left|z_k^{(t)} - s_i\right|\right)}
\]
Notice that the right-hand side of this expression is a function \textit{exclusively} of the current joint action profile $\mathbf{z}^{(t)}$ and the static, invariant intrinsic opinions $\mathbf{s} = (s_1, \dots, s_n)$. Crucially, it contains no functional dependence on:
\begin{enumerate}
    \item Past network topologies or actions: $\left(G^{(0)}, \mathbf{z}^{(0)}, \dots, G^{(t-1)}, \mathbf{z}^{(t-1)}\right)$
    \item The current network topology structure: $G^{(t)}$
\end{enumerate}

Because the sampling mechanism lacks memory and does not depend on the current graph ties, we can drop the historical terms and the current state from the conditioning set:
\begin{equation}\label{eq:indep}
P\left(S_i^{(t+1)} \;\middle|\; H^{(t)}, \mathbf{z}^{(t)}\right) = P\left(S_i^{(t+1)} \;\middle|\; \mathbf{z}^{(t)}\right)
\end{equation}

Substituting Equation \eqref{eq:indep} back into Equation \eqref{eq:factorization} yields:
\[
P\left(G^{(t+1)} \;\middle|\; H^{(t)}, \mathbf{z}^{(t)}\right) = \prod_{i=1}^n P\left(S_i^{(t+1)} \;\middle|\; \mathbf{z}^{(t)}\right) = P\left(G^{(t+1)} \;\middle|\; \mathbf{z}^{(t)}\right)
\]
By identical logic, conditioning the future state on only the current state $G^{(t)}$ and action $\mathbf{z}^{(t)}$ yields the exact same reduction:
\[
P\left(G^{(t+1)} \;\middle|\; G^{(t)}, \mathbf{z}^{(t)}\right) = P\left(G^{(t+1)} \;\middle|\; \mathbf{z}^{(t)}\right)
\]
Therefore, we obtain:
\[
P\left(G^{(t+1)} \;\middle|\; H^{(t)}, \mathbf{z}^{(t)}\right) = P\left(G^{(t+1)} \;\middle|\; G^{(t)}, \mathbf{z}^{(t)}\right)
\]
This completes the proof of the \textbf{Controlled Markov Property} (in fact, showing an even stronger property: the state transition depends purely on the actions chosen at that step).

\subsubsection*{Step~2: Proving the Autonomous Markov Chain Property.}
Let us now restrict our attention to the network topology sequence $\mathcal{G}^{(t)} = \left(G^{(0)}, G^{(1)}, \dots, G^{(t)}\right)$ under the assumption that players follow stationary Markov policies $\pi_i : \mathcal{S} \rightarrow \Delta(\mathcal{A}_i)$. Under such strategies, the probability of selecting an action profile depends strictly on the current state:
\[
P\left(\mathbf{z}^{(t)} \;\middle|\; \mathcal{G}^{(t)}\right) = P\left(\mathbf{z}^{(t)} \;\middle|\; G^{(t)}\right) = \prod_{i=1}^n \pi_i\left(z_i^{(t)} \;\middle|\; G^{(t)}\right)
\]

We evaluate the conditional distribution of the future state given only the history of states. By the Law of Total Probability, integrating (or summing) over the action space $\mathcal{A}$:
\[
P\left(G^{(t+1)} \;\middle|\; \mathcal{G}^{(t)}\right) = \int_{\mathcal{A}} P\left(G^{(t+1)} \;\middle|\; \mathcal{G}^{(t)}, \mathbf{z}^{(t)}\right) \cdot P\left(\mathbf{z}^{(t)} \;\middle|\; \mathcal{G}^{(t)}\right) d\mathbf{z}^{(t)}
\]
Using the conditional independence result from Step 1, $\mathbb{P}\left(G^{(t+1)} \;\middle|\; \mathcal{G}^{(t)}, \mathbf{z}^{(t)}\right) = \mathbb{P}\left(G^{(t+1)} \;\middle|\; \mathbf{z}^{(t)}\right)$. Substituting this and the Markov policy rule gives:
\begin{equation}\label{eq:chain_step}
P\left(G^{(t+1)} \;\middle|\; \mathcal{G}^{(t)}\right) = \int_{\mathcal{A}} P\left(G^{(t+1)} \;\middle|\; \mathbf{z}^{(t)}\right) \cdot P\left(\mathbf{z}^{(t)} \;\middle|\; G^{(t)}\right) d\mathbf{z}^{(t)}
\end{equation}

Now, let us examine the conditional probability of $G^{(t+1)}$ given \textit{only} the current state $G^{(t)}$:
\begin{align*}
P\left(G^{(t+1)} \;\middle|\; G^{(t)}\right) &= \int_{\mathcal{A}} P\left(G^{(t+1)} \;\middle|\; G^{(t)}, \mathbf{z}^{(t)}\right) \cdot P\left(\mathbf{z}^{(t)} \;\middle|\; G^{(t)}\right) d\mathbf{z}^{(t)} \\
&= \int_{\mathcal{A}} P\left(G^{(t+1)} \;\middle|\; \mathbf{z}^{(t)}\right) \cdot P\left(\mathbf{z}^{(t)} \;\middle|\; G^{(t)}\right) d\mathbf{z}^{(t)}
\end{align*}
Comparing this directly to Equation \eqref{eq:chain_step}, we observe that the integrals are mathematically identical. Hence:
\[
P\left(G^{(t+1)} \;\middle|\; G^{(0)}, G^{(1)}, \dots, G^{(t)}\right) = P\left(G^{(t+1)} \;\middle|\; G^{(t)}\right)
\]

\subsubsection*{Conclusion.}
The network topologies function as valid Markovian states. The continuous interaction loop operates as a classic Markov game because:
\begin{enumerate}
    \item The randomized transition mechanism to a new graph structure depends exclusively on the immediate configuration of expressed opinions.
    \item When closed under Markovian equilibrium policies, the structural transitions collapse into a pure Markov chain where historical topologies provide zero additional predictive power over the future evolution of the network.
\end{enumerate}

\subsection{Modeling the $K$-NN Markov Game} \label{sec:k-nn}

We formally define this coevolutionary setting as a \textbf{Markov Game} (or Stochastic Game) represented by the tuple:
\[
\mathcal{M} = \left\langle \mathcal{N}, \mathcal{S}, \{\mathcal{A}_i\}_{i \in \mathcal{N}}, \mathcal{P}, \{C_i\}_{i \in \mathcal{N}}, \gamma \right\rangle
\]

\subsubsection{Players ($\mathcal{N}$)}
The player set is $\mathcal{N} = \{1, 2, \dots, n\}$, corresponding to the $n$ nodes in the social network. Each player $i \in \mathcal{N}$ possesses an invariant, time-independent \textbf{intrinsic opinion} denoted by $s_i \in \mathbb{R}$.

\subsubsection{State Space ($\mathcal{S}$)}
The state space represents the finite set of valid \textbf{network topologies}. A state $G^{(t)} \in \mathcal{S}$ at time step $t$ is a directed graph profile characterized by the current neighborhood sets of all players:
\[
G^{(t)} = \left(S_1^{(t)}, S_2^{(t)}, \dots, S_n^{(t)}\right)
\]
Because each player must actively select exactly $K$ neighbors, every localized neighborhood set $S_i^{(t)} \subset \mathcal{N} \setminus \{i\}$ satisfies the strict cardinality constraint $|S_i^{(t)}| = K$.

\subsubsection{Action Space ($\mathcal{A}_i$)}
The action $a_i^{(t)}$ chosen by player $i$ at time step $t$ is their \textbf{expressed opinion}, denoted as $z_i^{(t)} \in \mathbb{R}$. The joint action profile of the population at time $t$ is given by the vector:
\[
\mathbf{z}^{(t)} = \left(z_1^{(t)}, z_2^{(t)}, \dots, z_n^{(t)}\right) \in \mathcal{A}
\]
where $\mathcal{A} = \prod_{i \in \mathcal{N}} \mathcal{A}_i$.

\subsubsection{Transition Probabilities ($\mathcal{P}$)}
Transitions characterize how the network topology alters state configurations from $t$ to $t+1$ based on the current action profile. Let $d_{ij}(z_j^{(t)}, s_i) = |z_j^{(t)} - s_i|$ represent the distance between player $j$'s expressed opinion and player $i$'s intrinsic opinion.

To model the randomized choice of ``roughly nearby'' neighbors, we apply an exponential soft-max weighting function:
\[
w_{ij}\left(z_j^{(t)}, s_i\right) = \exp\left(-\beta \cdot \left|z_j^{(t)} - s_i\right|\right)
\]
where $\beta > 0$ is the intensity of choice parameter (or inverse temperature) controlling the degree of randomness.

Since player $i$ samples a subset of $K$ distinct neighbors without replacement, the probability that player $i$ selects a specific neighborhood configuration $S_i^{(t+1)}$ is proportional to the product of its connection weights:
\[
P\left(S_i^{(t+1)} \;\middle|\; \mathbf{z}^{(t)}, s_i\right) = \frac{\prod_{j \in S_i^{(t+1)}} w_{ij}\left(z_j^{(t)}, s_i\right)}{\sum_{\substack{S' \subset \mathcal{N} \setminus \{i\} \\ |S'| = K}} \prod_{k \in S'} w_{ik}\left(z_k^{(t)}, s_i\right)}
\]
Assuming that players choose their neighbor portfolios conditionally independently given the current action profile, the global state transition probability function $\mathcal{P}: \mathcal{S} \times \mathcal{A} \rightarrow \Delta(\mathcal{S})$ factorizes as:
\[
P\left(G^{(t+1)} \;\middle|\; G^{(t)}, \mathbf{z}^{(t)}\right) = \prod_{i=1}^n P\left(S_i^{(t+1)} \;\middle|\; \mathbf{z}^{(t)}, s_i\right)
\]
\textit{Note:} In this specific class of games, the next topological state $G^{(t+1)}$ is conditionally independent of the historical state $G^{(t)}$ given the current opinions $\mathbf{z}^{(t)}$.

\subsubsection{Stage Cost Functions ($C_i$)}
The immediate stage cost $C_i$ sustained by player $i$ is calculated over the current state configuration ($G^{(t)}$) and the current joint action profile ($\mathbf{z}^{(t)}$). Adapting the continuous disagreement formulation from the base paper, the cost is defined as:
\[
C_i\left(G^{(t)}, \mathbf{z}^{(t)}\right) = \sum_{j \in S_i^{(t)}} \left(z_i^{(t)} - z_j^{(t)}\right)^2 + \rho K \left(z_i^{(t)} - s_i\right)^2
\]
where $S_i^{(t)}$ is the set of outgoing links designated to node $i$ by the current state graph $G^{(t)}$, and \(\rho > 0\) represents the internal penalty parameter weighting the player's aversion to drifting away from their true intrinsic opinion $s_i$.

\subsection{Optimistic Gradient Ascents in Two-Player General-Sum Markov Games} \label{sec:ogda}

Specifically in a two-player case, during iteration $t$, the algorithm defines game matrix $Q_{i,t}^s$ via $V_{i,t-1}^s$ for $i\in\{1,2\}$. The entries of $Q_{i,t}^s$ and $V_{i,t}^s$ can be formally expressed as follows: for all~$i,t,s$,
\begin{equation}
    Q_{i,t}^s(a, b) \triangleq r_i(s, a, b)+\gamma \mathbf{E}_{s^{\prime} \sim P(\cdot \mid s, a, b)}\left[V_{i,t-1}^{s^{\prime}}\right],
\end{equation}
\begin{equation}
    V_{i,t}^s=
    \sum_{a,b}Q_{i,t}^s(a, b).
\end{equation}

The OGA follows the updates below:
\begin{eqnarray}
    \widehat{\mathbf{x}}_{t+1}^s =\Pi_{\Delta_{\mathcal{X}}}\left\{\widehat{\mathbf{x}}_t^s+\eta \mathbf{u}_t^s\right\}, \qquad
    \mathbf{x}_{t+1}^s =\Pi_{\Delta_{\mathcal{X}}}\left\{\widehat{\mathbf{x}}_{t+1}^s+\eta \mathbf{u}_t^s\right\},\\ \nonumber
    \widehat{\mathbf{y}}_{t+1}^s =\Pi_{\Delta_{\mathcal{Y}}}\left\{\widehat{\mathbf{y}}_t^s+\eta \mathbf{r}_t^s\right\}, \qquad
    \mathbf{y}_{t+1}^s =\Pi_{\Delta_{\mathcal{Y}}}\left\{\widehat{\mathbf{y}}_{t+1}^s+\eta \mathbf{r}_t^s\right\},
\end{eqnarray}
where $\eta$ is a learning rate, and $\mathbf{u}_{t}^s$ is an estimator that represents the gradient at $\mathbf{x}_{t+1}^s$ and $\widehat{\mathbf{x}}_{t+1}^s$, and $\mathbf{r}_{t}^s$ is an estimator that represents the gradient at ${\mathbf{y}}_{t+1}^s$ and $\widehat{\mathbf{y}}_{t+1}^s$. Note that we adapt the estimation used in~\cite{wei2021last}~[Section~3.1] in zero-sum Markov games.
The estimator~$\mathbf{u}_{t}^s$ is defined as $\left\|\mathbf{u}^s_t-Q_{1,t}^s \mathbf{y}_{t}^s\right\| \leq \varepsilon$ to approximate $Q_{1,t}^s \mathbf{y}_{t}^s$. The estimator~$\mathbf{r}_{t}^s$ is defined as $\left\|\mathbf{r}^s_t-{\mathbf{x}_t^s}^\top Q_{2,t}^s \right\| \leq \varepsilon$ to approximate ${\mathbf{x}_t^s}^\top Q_{2,t}^s$. The estimation is detailed as follows.

\subsubsection*{Gradient Estimation.}
The algorithm uses these estimators as the (estimated) gradients instead of directly using $Q_{1,t}^s \mathbf{y}_t^s$ and ${\mathbf{x}_t^s}^\top Q_{2,t}^s$ because in \emph{decentralized} algorithms, the players can only observe their current state and reward after taking an action. Similarly, the updates for $\mathbf{y}_{t+1}^s$ and $\widehat{\mathbf{y}}_{t+1}^s$ are performed using the estimated gradient $\mathbf{r}_{t}^s$ instead of~$\mathbf{x}_t^s Q_{2,t}^s$.

In a decentralized algorithm, each agent can only access their current state and the reward received after taking an action. In this case, we compute the $\varepsilon$-approximations of $Q_{1,t}^s y_{t}^s$ and ${x_t^s}^\top Q_{2,t}^s$ as the gradients for $x^s$ and $y^s$, respectively. In particular, in each iteration $t$, both players participate in a sequence of $L$ steps, interacting with each other. In addition, they follow a policy that incorporates a \emph{uniformly random exploration}. 
\begin{eqnarray*}
&&u^{s}_t(a) =\frac{\sum_{i=1}^L \mathbf{1}\left[s_i=s, a_i=a\right]\left(r_1(s, a, b_i)+\gamma V_{1,t-1}^{s_{i+1}}\right)}{\sum_{i=1}^L \mathbf{1}\left[s_i=s, a_i=a\right]}\\
&&r^{s}_t(b) =\frac{\sum_{i=1}^L \mathbf{1}\left[s_i=s, b_i=b\right]\left(r_2(s, a_i, b)+\gamma V_{2, t-1}^{s_{i+1}}\right)}{\sum_{i=1}^L \mathbf{1}\left[s_i=s, b_i=b\right]}.\\
\end{eqnarray*}


\section{OGA for Approximate Nash Equilibria or Bounded POA in General-Sum Markov Games} \label{sec:oga}
We highlight the extra work and differences compared to the proof of \cite{anagnostides}[Theorem~A.17] on a high level as follows.
\begin{enumerate}
\item In (general-sum) Markov games, the regret of the optimistic gradient ascent algorithm would have extra positive terms (coming from Q-values) that need to be offset by part of the negative terms similar to those in the proof of \cite{anagnostides}[Theorem~A.12], from which the proof of \cite{anagnostides}[Theorem~A.17] follows. 

\item In our case, this results in assumption~(iii) required in \cite{anagnostides}[Theorem~A.17] unsatisfied since an appropriate range of $\eta$ (i.e., $\eta\leq O(\frac{1}{\sqrt{T}})$) has to be different from $\eta\leq\frac{C}{4C^*}$ in assumption~(iii) for two-player games, where $\|\cdot\|\geq C\|\cdot\|_1$ and $\|\cdot\|_*\leq C^*\|\cdot\|_\infty$ for some constants~$C$ and $C^*$. 

\item Furthermore, the threshold on $T$ for convergence or the bounded price of anarchy in assumption~(iii) of \cite{anagnostides}[Theorem~A.17] may not be satisfied or appropriate since that on $T$ is now different in our case (and is in relation to the number of states, which is not a parameter in the threshold on $T$ in assumption~(iii) of \cite{anagnostides}[Theorem~A.17]).
\end{enumerate}

We need some additional notions about Markov games. In a multi-player Markov game, there are ways to measure the convergence performance in terms of convergence rates. For any stationary policy profile $\mathbf{z}^s = (x^s_i)_{i}\in\mathcal{Z}^s=(\mathcal{X}_i^s)_i$ for all states~$s$, $dist(\mathbf{z}^s,\mathcal{Z}^s_*)$ is the distance from $\mathbf{z}^s$ to the set of
equilibrium policy profiles $\mathcal{Z}^s_*$ for $s$, which can be formulated as $\|\mathbf{z}^s - \Pi_{\mathcal{Z}^s_*}(\mathbf{z}^s)\|$, and we denote by $\Pi_{\mathcal{Z}^s_*}(\mathbf{z}) := \arg\min_{\mathbf{z}\in\mathcal{Z}^s_*} dist(\mathbf{z},\mathcal{Z}^s_*)$ the projection of $\mathbf{z}$ onto $\mathcal{Z}^s_*$. In particular, in the two-player case, we denote any stationary policy profile by $\mathbf{z}^s\in\mathcal{Z}^s = \mathcal{X}^s\times\mathcal{Y}^s$ for any $s$.

\begin{theorem} \label{thm:convergence}
Suppose that each player~$i$ employs OGA with (i) pair of norms $(\|\cdot\|, \|\cdot\|_*)$ such that $\|\cdot\|\geq C\|\cdot~\|_1$ and $\|\cdot\|\leq C^*\|\cdot~\|_\infty$; (ii)
$G_i$-smooth regularizer~$\mathcal{R}_i$; and (iii) learning rate $\eta\leq O(1/\sqrt{T})$. Moreover, suppose that the game is
such that $\sum_i R_i^T\geq 0$ for any $T$. Then, for any $\epsilon > 0$, after $T>\frac{2}{\epsilon^2 |\mathcal{S}|}\sum_{s_t\in\mathcal{S}}\sum_i\Omega_i^s$ iterations, there exists an iterate~$\mathbf{z}_t^s$ with $t\in [T]$ such that $\mathbf{z}^{s_t}_t$ is an
\[\epsilon\Bigl(C^*\max\{\|Q^{s_t}_{x,t}\mathbf{y}^{s_t}_t\|_\infty,\|\mathbf{x}^{s_t}_t Q^{s_t}_{y,t}\|_\infty\}+2\frac{\max\{G_1\Omega^s_1,G_2\Omega^s_2\}}{\eta}\Bigr)\] 
-approximate Nash equilibrium joint policy, where \\$\Omega^s_i:=\sup_{\mathbf{x},\mathbf{x}'\in\mathcal{X}^s_i}\|\mathbf{x}-~\mathbf{x}'\|$. 
Otherwise, \[
\frac{1}{T} \sum^T_{t=1} V_{\pi_t}(s) \geq \frac{\alpha}{1+\beta}V_{\pi^*}(s) + \frac{1}{1+\beta}\frac{\epsilon^2}{8 \eta}.
\] 
\end{theorem}
\begin{remark}
Although this result is stated for two players, it can be readily extended to multiple players. 
\end{remark}
\begin{proof}
By the standard proof of OGDA (see, e.g., the proof of Lemma~1 in \cite{wei2021linear} or \cite[Lemma~1]{rakhlin:sridharan}), for each~$t$ we have 
\begin{eqnarray*}
(\mathbf{x}^s_*-\mathbf{x}^s_t)^T \mathbf{u}^s_t&\leq&\frac{1}{2\eta}(\|\hat{\mathbf{x}}^s_t-\mathbf{x}^s_*\|^2-\|\hat{\mathbf{x}}^s_{t+1}-\mathbf{x}^s_*\|^2\\
&&-\|\hat{\mathbf{x}}^s_{t+1}-\mathbf{x}^s_t\|^2-\|\mathbf{x}^s_t-\hat{\mathbf{x}}^s_t\|^2)\\
&&+\eta\|\mathbf{u}^s_t-\mathbf{u}^s_{t-1}\|^2
\end{eqnarray*}
for any $s$. Thus, for each~$t$ 
\begin{eqnarray*}
\sum_s(\mathbf{x}^s_*-\mathbf{x}^s_t)^T \mathbf{u}^s_t&\leq&\sum_s(\frac{1}{2\eta}(\|\hat{\mathbf{x}}^s_t-\mathbf{x}^s_*\|^2-\|\hat{\mathbf{x}}^s_{t+1}-\mathbf{x}^s_*\|^2\\
&&-\|\hat{\mathbf{x}}^s_{t+1}-\mathbf{x}^s_t\|^2-\|\mathbf{x}^s_t-\hat{\mathbf{x}}^s_t\|^2)\\
&&+\eta\|\mathbf{u}^s_t-\mathbf{u}^s_{t-1}\|^2)
\end{eqnarray*}
Summing over $t$, by telescoping
\begin{eqnarray*}
\sum_t\sum_s(\mathbf{x}^{s}_*-\mathbf{x}^{s}_t)^{\top} \mathbf{u}^{s}_t&\leq&\frac{1}{2\eta}(\sum_s(\|\hat{\mathbf{x}}^{s}_1-\mathbf{x}^{s}_*\|^2-
\|\hat{\mathbf{x}}^{s}_{T+1}-\mathbf{x}^{s}_*\|^2)\\
&&-\sum_t\sum_s\|\hat{\mathbf{x}}^{s}_{t+1}-\mathbf{x}^{s}_t\|^2-\sum_t\sum_s\|\mathbf{x}^{s}_t-\hat{\mathbf{x}}^{s}_t\|^2)\\
&&+\eta\sum_t\sum_s\|\mathbf{u}^{s}_t-\mathbf{u}^{s}_{t-1}\|^2.
\end{eqnarray*}
By the definition of $\mathbf{u}^s_t$, we have
\[\eta\|\mathbf{u}^s_t-\mathbf{u}^s_{t-1}\|^2\leq 4\eta\|Q^s_{\mathbf{x},t}-Q^s_{\mathbf{x},t-1}\|^2+\frac{4\eta}{(1-\gamma)^2}\|\mathbf{y}^s_t-\mathbf{y}^s_{t-1}\|^2.\]
for any $s$.
Similarly, 
\begin{eqnarray*}
&&\sum_t\sum_s(\mathbf{y}^{s}_*-\mathbf{y}^{s}_t)^\top \mathbf{r}^{s}_t\\
&\leq&\frac{1}{2\eta}(\sum_s(\|\hat{\mathbf{y}}^{s}_1-\mathbf{y}^{s}_*\|^2-
\|\hat{\mathbf{y}}^{s}_{T+1}-\mathbf{y}^{s}_*\|^2) \\
&&-\sum_t\sum_s\|\hat{\mathbf{y}}^{s}_{t+1}-\mathbf{y}^{s}_t\|^2-\sum_t\sum_s\|\mathbf{y}^{s}_t-\hat{\mathbf{y}}^{s}_t\|^2)+\eta\sum_t\|\mathbf{r}^{s}_t-\mathbf{r}^{s}_{t-1}\|^2,
\end{eqnarray*} 
and
\begin{eqnarray*}
\eta\|\mathbf{r}^s_t-\mathbf{r}^s_{t-1}\|^2\leq4\eta\|Q^s_{\mathbf{y},t}-Q^s_{\mathbf{y},t-1}\|^2+\frac{4\eta}{(1-\gamma)^2}\|\mathbf{x}^s_t-\mathbf{x}^s_{t-1}\|^2
\end{eqnarray*}
for any $s$.

\subsection{Q-values} \label{sec:q}
To deal with extra positive terms that come from Q-values, we then borrow some technical auxiliary lemma that relates O-value differences to policy differences.
\begin{lemma}[Appendix~D of \cite{wei2021last}]
For any $s$,
\[\|Q_{\mathbf{x},t}^s-Q_{\mathbf{x},t-1}^s\|^2\leq\frac{8\gamma^2}{(1-\gamma)^3}\max_{s'}\sum_{\tau=1}^{t-1}\beta^\tau_t\|\mathbf{y}^{s'}_\tau-\mathbf{y}^{s'}_{\tau-1}\|^2\]
where $\gamma$ is defined in \cite[Appendix~D]{wei2021last}. 
\end{lemma}
Thus, for any $s$
\begin{eqnarray*}
&&\sum_{t=1}^T\|Q_{\mathbf{x},t}^{s}-Q_{\mathbf{x},t-1}^{s}\|^2\\
&\leq&\frac{8\gamma^2}{(1-\gamma)^3}\sum_{t=1}^T\max_{s'}\sum_{\tau=1}^{t-1}\beta^\tau_t\|\mathbf{y}^{s'}_\tau-\mathbf{y}^{s'}_{\tau-1}\|^2\\
&\leq&\frac{8\gamma^2}{(1-\gamma)^3}\max_{s'}\sum_{t=1}^T\sum_{\tau=1}^{t-1}\beta^\tau_t\|\mathbf{y}^{s'}_\tau-\mathbf{y}^{s'}_{\tau-1}\|^2\\
&=&\frac{8\gamma^2}{(1-\gamma)^3}\max_{s'}((T-1)\max_{t\in\{2,...,T\}}\beta^1_t\|\mathbf{y}_1^{s'}-\mathbf{y_0}^{s'}\|+...+\max_{t\in\{2,...,T\}}\beta^{T-1}_t\|\mathbf{y}_{T-1}^{s'}-\mathbf{y}_{T-2}^{s'}\|^2)\\
&\leq&\frac{8\gamma^2 T\max_{t\in\{2,...,T\},\tau(t)\in\{1,...,t-1\}}\beta^{\tau(t)}_t}{(1-\gamma)^3}\cdot\max_{s'}\sum_{t=1}^{T-1}\|\mathbf{y}_t^{s'}-\mathbf{y}_{t-1}^{s'}\|^2.
\end{eqnarray*}

We are ready to bound the regrets. Let $\mathbf{z}=(\mathbf{x},\mathbf{y})$. Summing over the regrets of two players, 
\begin{eqnarray*}
&&R_1^T+R_2^T\\
&\leq&\frac{1}{2\eta}(\sum_s\|\hat{\mathbf{z}}^{s}_1-\mathbf{z}^{s}_*\|^2-\sum_t\sum_s\|\hat{\mathbf{z}}^{s}_{t+1}-\mathbf{z}^{s}_t\|^2-\sum_t\sum_s\|\mathbf{z}^{s}_t-\hat{\mathbf{z}}^{s}_t\|^2)\\
&&+4\eta\sum_t\sum_s(\|Q^{s}_{x,t}-Q^{s}_{x,t-1}\|^2+\|Q^{s}_{y,t}-Q^{s}_{y,t-1}\|^2\\
&&+\frac{1}{(1-\gamma)^2}\|\mathbf{z}^{s}_t-\mathbf{z}^{s}_{t-1}\|^2)
\end{eqnarray*}
for any $s$. Thus,

\begin{eqnarray*}
&&R_1^T+R_2^T\\
&\leq&\frac{1}{2\eta}\sum_s\|\hat{\mathbf{z}}^s_1-\mathbf{z}^s_*\|^2\\
&&-\frac{1}{4\eta|\mathcal{S}|}(\sum_t\sum_s\|\hat{\mathbf{z}}^{s_t}_{t+1}-\mathbf{z}^{s_t}_t\|^2+\sum_t\sum_s\|\mathbf{z}^{s_t}_t-\hat{\mathbf{z}}^{s_t}_t\|^2)\\
&&-\frac{1}{4\eta}\min_{s'}(\sum_t\|\hat{\mathbf{x}}^{s'}_{t+1}-\mathbf{x}^{s'}_t\|^2+\sum_t\|\mathbf{x}^{s'}_t-\hat{\mathbf{x}}^{s'}_t\|^2)\\
&&-\frac{1}{4\eta}\min_{s'}(\sum_t\|\hat{\mathbf{y}}^{s'}_{t+1}-\mathbf{y}^{s'}_t\|^2+\sum_t\|\mathbf{y}^{s'}_t-\hat{\mathbf{y}}^{s'}_t\|^2)\\
&&+\frac{4\eta}{|\mathcal{S}|}\sum_t\sum_s(\|Q^{s_t}_{x,t}-Q^{s_t}_{x,t-1}\|^2+\|Q^{s_t}_{y,t}-Q^{s_t}_{y,t-1}\|^2) \\
&&+\frac{4\eta}{(1-\gamma)^2}\max_{s'}\sum_t\|\mathbf{x}^{s'}_t-\mathbf{x}^{s'}_{t-1}\|^2\\
&&+\frac{4\eta}{(1-\gamma)^2}\max_{s'}\sum_t\|\mathbf{y}^{s'}_t-\mathbf{y}^{s'}_{t-1}\|^2,
\end{eqnarray*}
where the inequality holds by Young's inequality and the triangle inequality,
\[\|\mathbf{z}^s_{t+1}-\mathbf{z}^s_{t}\|^2\leq 2\|\mathbf{z}^s_{t+1}-\hat{\mathbf{z}}^s_{t}\|^2+2\|\mathbf{z}^s_{t}-\hat{\mathbf{z}}^s_{t}\|^2,\] and summing over all $t \in [T]$,
\begin{eqnarray*}
&&\sum_{t=1}^T\|\mathbf{z}^s_{t+1}-\mathbf{z}^s_{t}\|^2\\
& \leq & 2\sum_t\|\mathbf{z}^s_{t+1}-\hat{\mathbf{z}}^s_{t}\|^2+2\sum_t\|\mathbf{z}^s_{t} - \hat{\mathbf{z}}^s_{t}\|^2\leq 2\sum_t\|\hat{\mathbf{z}}^s_{t+1}-\mathbf{z}^s_t\|^2+2\sum_t\|\mathbf{z}^s_t-\hat{\mathbf{z}}^s_t\|^2.
\end{eqnarray*}

Using this, we can find a proper range for $\eta$ from player~$x$'s perspective,
\begin{eqnarray*}
&&\Bigl(-\frac{1}{8\eta}+\frac{32\eta\gamma^2 T\max_{t\in\{2,...,T\},\tau\in\{1,...,T-1\}}\beta^\tau_t}{(1-\gamma)^3}+\frac{4\eta}{(1-\gamma)^2}\Bigr)\cdot\max_{s'}\sum_{t=1}^T \|\mathbf{y}_{t+1}^{s'}-\mathbf{y}_{t}^{s'}\|^2\leq 0,
\end{eqnarray*}
and similarly for player~$y$ so for $\eta\leq O(1/\sqrt{T})$, 
\begin{eqnarray*}
R_1^T+R_2^T&\leq&\frac{1}{2\eta}\sum_s\|\hat{\mathbf{z}}^s_1-\mathbf{z}^s_*\|^2-\frac{1}{4\eta|\mathcal{S}|}\Bigl(\sum_t\sum_s\|\hat{\mathbf{z}}^s_{t+1}-\mathbf{z}^s_t\|^2+\sum_t\sum_s\|\mathbf{z}^s_t-\hat{\mathbf{z}}^s_t\|^2\Bigr).
\end{eqnarray*}

\subsection{Approximate Nash Equilibria or Bounded POA}
For $\eta\leq O(1/\sqrt{T})$, we have
\begin{eqnarray*}
0&\leq&\frac{1}{2\eta|\mathcal{S}|}\sum_s\|\hat{\mathbf{z}}^s_1-\mathbf{z}^s_*\|^2-\frac{1}{4\eta|\mathcal{S}|}(\sum_t\sum_s\|\hat{\mathbf{z}}^s_{t+1}-\mathbf{z}^s_t\|^2+\sum_t\sum_s\|\mathbf{z}^s_t-\hat{\mathbf{z}}^s_t\|^2),
\end{eqnarray*}
implying that
\begin{eqnarray} \label{eq:ineq} \nonumber
&&(\sum_t\sum_s\|\hat{\mathbf{z}}^s_{t+1}-\mathbf{z}^s_t\|^2+\sum_t\sum_s\|\mathbf{z}^s_t-\hat{\mathbf{z}}^s_t\|^2)
\leq\sum_s\|\hat{\mathbf{z}}^s_1-\mathbf{z}^s_*\|^2.
\end{eqnarray}

We adapt and extend the analysis idea in~\cite[Theorem~A.12 and Claim~A.14]{anagnostides} under MDP.
\begin{itemize}
\item Assume that for all $t\in[T]$, it holds that if there exists $t\in[T]$, $\|\hat{\mathbf{z}}^s_{t+1}-\mathbf{z}^s_t\|^2+\|\mathbf{z}^s_t-\hat{\mathbf{z}}^s_t\|^2>\epsilon^2$. In this case, it follows from~(\ref{eq:ineq}) that
\begin{eqnarray*}
\epsilon^2 T|\mathcal{S}|\leq 2\sum_s\|\hat{\mathbf{z}}^s_1-\mathbf{z}^s_*\|^2\implies 
T\leq\frac{2}{\epsilon^2 |\mathcal{S}|}\sum_s\|\hat{\mathbf{z}}^s_1-\mathbf{z}^s_*\|^2.
\end{eqnarray*}
Thus, for $T>2\sum_s\|\hat{\mathbf{z}}^s_1-\mathbf{z}^s_*\|^2/(\epsilon^2 |\mathcal{S}|)$ it must be the case that there exists $t\in [T]$ such that  
\begin{eqnarray} \label{eq:singlestep}
\|\hat{\mathbf{z}}^s_{t+1}-\mathbf{z}^s_t\|^2+\|\mathbf{z}^s_t-\hat{\mathbf{z}}^s_t\|^2\leq\epsilon^2.
\end{eqnarray}

This implies that 
\begin{enumerate}
    \item $\|\mathbf{x}^s_t-\hat{\mathbf{x}}^s_t\|\leq\epsilon$ or $\|\mathbf{y}^s_t-\hat{\mathbf{y}}^s_t\|\leq\epsilon$; 
    \item $\|\hat{\mathbf{z}}^s_{t+1}-\hat{\mathbf{z}}^s_t\|^2\leq 2\|\hat{\mathbf{z}}^s_{t+1}-\mathbf{z}^s_t\|^2+2\|\mathbf{z}^s_t-\hat{\mathbf{z}}^s_t\|^2 \implies \|\hat{\mathbf{z}}^s_{t+1}-\hat{\mathbf{z}}^s_t\|\leq 2\epsilon$ by the update rule of OGA ~(\ref{eq:singlestep}).
\end{enumerate}

Finally, we derive the following result.
\begin{claim} \label{thm:result}
Then, it follows that $\mathbf{z}^s_t$ is an
\begin{eqnarray*}
\epsilon\Bigl(C^*\max\{\|Q^s_{x,t}\mathbf{y}^s_t\|_\infty,\|\mathbf{x}^s_t Q^s_{y,t}\|_\infty\}+2\frac{\max\{G_1\Omega^s_1,G_2\Omega^s_2\}}{\eta}\Bigr)
\end{eqnarray*}
approximate Nash equilibrium policies, where $\|\cdot\|\leq C^*\|\cdot~\|_\infty$, 
regularizer~$\mathcal{R}_i$ is $G_i$-smooth and $\Omega^s_i:=\\\sup_{\mathbf{x},\mathbf{x}'\in\mathcal{X}_i^s}\|\mathbf{x}-~\mathbf{x}'\|$.
\end{claim}

Observe that the maximization problem associated with OMD, which becomes OGA with a square $\ell_2$-norm regularizer,  can be expressed in the following variational inequality form: (Regularizer $\mathcal{R}(v) = \frac{1}{2}\|v\|^2$ in OGA)
\[\langle \mathbf{u}^s_t-\frac{1}{\eta}(\nabla\mathcal{R}(\hat{\mathbf{x}}^s_t)-\nabla\mathcal{R}_{\mathbf{x}}(\hat{\mathbf{x}}^s_{t-1}),\hat{\mathbf{x}}^s-\hat{\mathbf{x}}^s_t\rangle\leq 0,\forall\hat{\mathbf{x}}^s\in\mathcal{Z}^s.\]
Thus, it follows that
\begin{eqnarray} \label{eq:bound}
\langle \mathbf{u}^s_t,\hat{\mathbf{x}}^s-\hat{\mathbf{x}}^s_t\rangle\leq 2\epsilon\frac{G_1\Omega^s_1}{\eta}.
\end{eqnarray}

Moreover, we also have that 
\[|\langle \mathbf{u}^s_t,\mathbf{x}^s-\hat{\mathbf{x}}^s_t\rangle| \leq \|\mathbf{u}^s_t\|_*\|\mathbf{x}^s-\hat{\mathbf{x}}^s_t\| \leq \epsilon C^*\|\mathbf{u}^s_t\|_\infty\]
Combining (\ref{eq:bound}) gives us that for any $\hat{\mathbf{x}}^s\in\mathcal{Z}_1^s$
\begin{eqnarray*}
\langle \mathbf{x}^s_t,\mathbf{u}^s_t\rangle &\geq&\langle\hat{\mathbf{x}}^s_t,\mathbf{u}^s_t\rangle-\epsilon C^*\|\mathbf{u}^s_t\|_\infty\\
&\geq&\langle\hat{\mathbf{x}}^s,\mathbf{u}^s_t\rangle-\epsilon C^*\|\mathbf{u}^s_t\|_\infty-2\epsilon\frac{G_1\Omega^s_1}{\eta}.
\end{eqnarray*}
Similarly, for $\langle \mathbf{y}^s_t,\mathbf{r}^s_t\rangle$ we have the corresponding inequality.

\item Otherwise, $\|\hat{\mathbf{z}}^s_{t+1}-\mathbf{z}^s_t\|^2+\|\mathbf{z}^s_t-\hat{\mathbf{z}}^s_t\|^2>\epsilon^2$ for all~$t$,
\begin{eqnarray} \label{eq:total reg bound}
R_1^T+R_2^T&\leq&\frac{1}{2\eta|\mathcal{S}|}\sum_s\|\hat{\mathbf{z}}^s_1-\mathbf{z}^s_*\|^2-\sum_{t=1}^T\frac{\epsilon^2}{4\eta}\nonumber\\ 
&\leq& -\frac{1}{8\eta}\epsilon^2 T
\end{eqnarray}
as long as $T \geq 4\sum_s\|\hat{\mathbf{z}}^s_1-\mathbf{z}^s_*\|^2/(\epsilon^2|\mathcal{S}|)$. 
Since $R_1^T+R_2^T$ is the total regret of all players from time $1$ to~$T$ in our Markov games, we may conclude from (\ref{eq:total reg bound}) that
\[
\frac{1}{T} \sum^T_{t=1} V_{\pi_t}(s) \geq \frac{\alpha}{1+\beta}V_{\pi^*}(s) + \frac{1}{1+\beta}\frac{\epsilon^2}{8 \eta}.
\] 
\end{itemize}
\end{proof}

\section{Conclusions and Future Work}
A stronger sense of convergences, for instance, last-iterate convergence, than our current one is possible and left for future work. 

Furthermore, we can generalize \cite{anagnostides}[Theorem~A.17] with relaxed assumptions: when the regret has some slightly more general form as in general-sum Markov games in our case, we still can have results of convergence to approximate Nash equilibria or bounds on the price of anarchy. 
\cite{anagnostides}[Theorem~A.17] and our result of Theorem~\ref{thm:convergence} are both special cases.

\subsubsection*{Acknowledgements.} 
We would like to thank Wei-Chen Lin for useful discussions.

%
%
\bibliographystyle{splncs04}
%
\bibliography{multiagent, references}

\appendix







\section{Appendix} \label{app:a}

\begin{definition}[Mixed Nash Equilibrium]
A vector~$(\sigma_1,..., \sigma_n)$ of independent
probability distributions over strategy (or policy in Markov games) sets is a mixed-strategy (or policy in Markov games) Nash equilibrium if no
player can improve her payoff under the product distribution $\mathbf{\sigma} = \sigma_1\times\cdots\times\sigma_n$ via a
unilateral deviation: $\mathbf{E}_{\mathbf{s}\sim\mathbf{\sigma}}[p_i(\mathbf{s})] \geq \mathbf{E}_{\mathbf{s}_{-i}\sim\mathbf{\sigma}_{-i}}[p_i(s'_i,\mathbf{s}_{-i})]$ for all $i$ and $s'_i \in S_i$, where $\sigma_{-i}$ is the product distribution of all $\sigma_j$'s other than $\sigma_i$, $p_i$ is player~$i$'s payoff, and $S_i$ is player~$i$'s strategy (or policy in Markov games) space. \end{definition}

\begin{definition}[Correlated Equilibrium] A correlated equilibrium is a joint probability distribution $\mathbf{\sigma} = \sigma_1\times...\times\sigma_n$ over the strategy (or policy in Markov games) profiles with the property that $\mathbf{E}_{\mathbf{s}\sim\mathbf{\sigma}}[p_i(\mathbf{x})|s_i] \geq \mathbf{E}_{\mathbf{s}\sim\mathbf{\sigma}}[p_i(s'_i,\mathbf{s}_{-i})|s_i]$ for all $i$ and $s_i, s'_i \in S_i$, where $\sigma_\mathbf{s}$ is the probability that the strategy (or policy in Markov games) profile is $\mathbf{s}$, $p_i$ is player~$i$'s payoff, and $S_i$ is player~$i$'s strategy (or policy in Markov games) space.
\end{definition}

\begin{definition}[Coarse Correlated Equilibrium]
A coarse correlated equilibrium
or coarse equilibrium is a joint probability distribution $\mathbf{\sigma}$ over strategy (or policy in Markov games) profiles that satisfy
$\mathbf{E}_{\mathbf{s}\sim\sigma}[p_i(\mathbf{s})] \geq \mathbf{E}_{\mathbf{s}\sim\mathbf{\sigma}}[p_i(s'_i,\mathbf{s}_{-i})]$ for every $i$ and $s_i, s'_i \in S_i$, where $p_i$ is player~$i$'s payoff, and $S_i$ is player~$i$'s strategy (or policy in Markov games) space.
\end{definition}



\end{document}